\pdfoutput=1
\RequirePackage{fix-cm}
\documentclass[a4paper,UKenglish]{lipics}
\hypersetup{%
   breaklinks,%
   colorlinks=true,%
   linkcolor=[rgb]{0.45,0.0,0.0},%
   urlcolor=[rgb]{0.05,0.390,0.0},%
   citecolor=[rgb]{0,0,0.45}
}%

\usepackage{microtype}
\usepackage{lmodern} 

\usepackage{amssymb,latexsym,amsthm}
\usepackage{adjustbox}
\usepackage{algorithm}
\usepackage{enumerate}
\usepackage{appendix}
\usepackage{amsmath}
\usepackage{breqn}
\usepackage[numbers]{natbib}

\usepackage{mathtools}

\usepackage{etoolbox}
\makeatletter
    \pretocmd{\NAT@citexnum}{\@ifnum{\NAT@ctype>\z@}{\let\NAT@hyper@\relax}{}}{}{}
\makeatother

\usepackage{xspace} 
\usepackage{color}%

\usepackage{subcaption}
\usepackage{pgf,tikz}
\usepackage{wrapfig}  
\usepackage{cutwin}  

\usepackage[nameinlink]{cleveref} 
\usepackage[noend]{algorithmic}
\Crefname{ALC@unique}{Line}{Lines} 

\crefname{ineq}{inequality}{inequalities}
\creflabelformat{ineq}{#2{\upshape(#1)}#3} 
\crefname{claim}{claim}{claims}           
\crefname{defn}{definition}{definitions}  
\crefname{obs}{Observation}{Observations}
\crefname{lem}{Lemma}{Lemmas}
\crefname{cor}{Corollary}{Corollaries}

\usetikzlibrary{arrows,shapes,calc,intersections,through,backgrounds,patterns}
\usepackage{thmtools, thm-restate}

\theoremstyle{plain}

\newtheorem{claim}{Claim}[section]
\newtheorem{defn}{Definition}[section]
\newtheorem{obs}[defn]{Observation}

\newcommand{\eps}{\ensuremath{\varepsilon}}

\newcommand{\V}{\ensuremath{\cal V}}

\newcommand{\Z}{\ensuremath{\cal Z}}
\newcommand{\Op}{\ensuremath{\cal O}}
\newcommand{\Lp}{\ensuremath{\cal L}}

\begin{document}
\title{On Variants of $k$-means Clustering\footnote{This material is based upon work supported by
the National Science Foundation under Grant CCF-1318996}}

\author{Sayan Bandyapadhyay\thanks{sayan-bandyapadhyay@uiowa.edu}} 
\author{Kasturi Varadarajan\thanks{kasturi-varadarajan@uiowa.edu}}
\affil{
  Department of Computer Science\\
  University of Iowa, Iowa City, USA\\
  }
  
\authorrunning{S.\,Bandyapadhyay and K.\,Varadarajan} 

\Copyright{Sayan Bandyapadhyay and Kasturi Varadarajan}

\keywords{$k$-means, Facility location, Local search, Geometric approximation}
\maketitle

\begin{abstract}
\textit{Clustering problems} often arise in the fields like data mining, machine learning etc. to group a collection of objects into similar groups with respect to a similarity (or dissimilarity) measure. Among the clustering problems, specifically \textit{$k$-means} clustering has got much attention from the researchers. Despite the fact that $k$-means is a very well studied problem its status in the plane is still an open problem. In particular, it is unknown whether it admits a PTAS in the plane. The best known approximation bound in polynomial time is $9+\eps$.

In this paper, we consider the following variant of $k$-means. Given a set $C$ of points in $\mathcal{R}^d$ and a real $f > 0$, find a finite set $F$ of points in $\mathcal{R}^d$ that minimizes the quantity $f*|F|+\sum_{p\in C} \min_{q \in F} {||p-q||}^2$. For any fixed dimension $d$, we design a local search PTAS for this problem. We also give a ``bi-criterion'' local search algorithm for $k$-means which uses $(1+\eps)k$ centers and yields a solution whose cost is at most $(1+\eps)$ times the cost of an optimal $k$-means solution. The algorithm runs in polynomial time for any fixed dimension.  

The contribution of this paper is two fold. On the one hand, we are being able to handle the square of distances in an elegant manner, which yields near optimal approximation bound. This leads us towards a better understanding of the $k$-means problem. On the other hand, our analysis of local search might also be useful for other geometric problems. This is important considering that very little is known about the local search method for geometric approximation.
\end{abstract} 

\section{Introduction}\label{sec:intro}
Given a set of items/objects, the \textit{clustering} problem is to group them into similar groups with respect to a similarity (or dissimilarity) measure. Due to its fundamental nature clustering has several applications in the fields like data mining, machine learning, pattern recognition, image processing and so on \cite{BroderGMZ97,DeerwesterDLFH90,duda1973pattern,FaloutsosBFHNPE94,FayyadPSU96,JainDM00,JainMF99,KaufmanR90,SwainB91}. Often the objects to cluster are mapped to a high dimensional metric space and the distance between two objects represents the similarity (or dissimilarity) between the objects. Then the goal is to minimize (or maximize) certain objective function that depends on the distances between the objects. Among the different variants of the clustering problem, specifically the \textit{$k$-means} problem has got much attention. In $k$-means clustering, given a set $P$ of $n$ points in $\mathcal{R}^d$ and an integer $k > 0$, the goal is to find a set $K$ of $k$ centers in $\mathcal{R}^d$, such that the quantity $$cost(K)=\sum_{p\in P} \min_{q \in K} ||p-q||^2$$ is minimized. $k$-means is known to be $\mathcal{NP}$-hard even in the plane \cite{AloiseDHP09,MahajanNV12}. Hence, there has been a lot of work to approximate the $k$-means objective function in polynomial time. There are even $(1+\eps)$-factor approximation algorithms whose time complexity depend linearly on $n$ \cite{Har-PeledK07,Har-PeledM04,KumarSS05,DMatousek00}. Unfortunately, the time complexity of these algorithms depend exponentially on $k$ and hence they are not suitable in practice when $k$ is sufficiently large. For arbitrary $k$ and $d$, the best known approximation factor is $9+\eps$ based on a local search technique \cite{KanungoMNPSW04}. On the other hand, for fixed dimension $d$, there is a bi-criteria approximation algorithm for $k$-means that uses $\beta k$ ($ > k$) centers and achieves an approximation factor $\alpha(\beta)$ $(<9+\eps) $ that depends on $\beta$ \cite{DMakarychevMSW15}. Moreover, $\alpha(\beta)$ decreases rapidly with $\beta$ (for instance $\alpha(2) <2.59,\alpha(3) <1.4$). Recently, Awasthi~{\em et. al} \cite{AwasthiCKS15} have attempted to study the inapproximabilty of $k$-means. They have shown that $k$-means is $APX$-hard in sufficiently high ($\Omega(\log n)$) dimensions. However, as they have pointed out in their paper, the status of this problem in constant dimensions is still not resolved. See also \cite{AwasthiBS10,OstrovskyRSS12} for some related work.


An insight about the difficulty of $k$-means can be found by comparing it to the \textit{$k$-median} clustering. $k$-median is similar to $k$-means except the goal is to minimize the sum of distances, instead of the sum of squares of distances. Arora~{\em et. al} \cite{Arora} presented a PTAS for $k$-median in plane based on a novel technique due to Arora \cite{Arora98}. Kolliopoulos and Rao \cite{KolliopoulosR07} improved the time complexity significantly to $O(\rho n\log n\log k)$, where $\rho = exp[O((1+\log 1/\eps)/\eps)^{d-1}]$ and $\eps$ is the constant of the PTAS. 
From the results on $k$-median one might conclude, that a reason behind the sophistication of $k$-means is its objective function. One reason that squares of distances are harder to handle compare to the distances is they do not follow triangle inequality in general. However, the important question in this context is, is the objective function of $k$-means by itself good enough to make this problem harder? We are interested to address this question in this paper. To set up the stage we consider another famous problem, which is called the facility location problem.

Facility location is similar to the $k$-median problem, where we are given a set of points (clients) in $\mathcal{R}^d$. The plan is to choose another set of points (facilities) in $\mathcal{R}^d$ which ``serve'' the clients. Though there is no global constraint on the number of facilities, for each facility, we need to pay a fixed cost. Here the objective function to minimize is the facility costs plus the sum of the distances from each of the clients to its nearest facility. Facility location has got much attention by the researchers in the fields like operation research, approximation algorithm etc. One can get a PTAS for facility location in the plane using the same technique by Arora~{\em et. al} \cite{Arora} that solves $k$-median. Actually, $k$-median has always been considered harder compare to facility location due to the global constraint on the number of centers as mentioned in \cite{Arora}. Now going back to our original question for $k$-means one can infer that the global constraint in $k$-means might as well play a crucial role. Motivated by this, we define the following variant of facility location.
\\\\\textbf{Sum of Squares Facility Location Problem (SOS-FL).} Given a set $C$ of points (clients) in $\mathcal{R}^d$ and a real $f > 0$, find a finite set $F$ of points (facilities) in $\mathcal{R}^d$, such that the quantity $$cost(F)=f*|F|+\sum_{p\in C} \min_{q \in F} {||p-q||}^2$$ is minimized.

Note that SOS-FL is similar to $k$-means except the global constraint on the number of facilities (or centers) is absent here. In this paper we study the following interesting question.

\begin{description}
 \item Is it possible to get a PTAS for SOS-FL in $\mathcal{R}^d$ for fixed $d$?
\end{description}

We answer this question positively. In particular for any $\eps> 0$, we give a $(1+\eps)$-factor approximation for SOS-FL based on a \textit{local search heuristic}. This result is very interesting, as it also addresses our earlier question regarding $k$-means. To be precise it infers, that it is the joint effect of the global constraint and the objective function that makes $k$-means complicated to deal with.
\\\\\textbf{Local Search.} Local search is a very popular heuristic in combinatorial optimization. But the technique was not much in use for geometric problems until recently. Still we know very little regarding this technique for geometric approximation. Arya~{\em et. al} \cite{AryaGKMMP04} gave a $3+\frac{2}{p}$ factor approximation for $k$-median based on a local search that swaps $p$ facilities, which was later simplified by Gupta and Tangwongsan \cite{abs-0809-2554}. The $9+\eps$ factor approximation for $k$-means, as mentioned before, is based on the approach of Arya~{\em et. al} \cite{AryaGKMMP04}. Mustafa and Ray \cite{MustafaR09} gave a local search PTAS for the discrete hitting set problem over pseudodisks and $r$-admissible regions in the plane. Chan and Har-Peled \cite{ChanH12} designed a local search heuristic for the independent set problem over fat objects, and for pseudodisks in the plane, which yields a PTAS. Recently, Cohen-Addad and Mathieu \cite{Cohen-AddadM15} showed the effectiveness of local search technique for geometric optimization by designing local search algorithms for many geometric problems including $k$-median and facility location. For facility location they achieved a PTAS. For $k$-median their approach yields a $1+\eps$ factor approximation by using at most $(1+\eps)k$ centers. Very recently, Bhattiprolu and Har-Peled \cite{BhattiproluH14} designed a local search PTAS for a geometric hitting set problem over balls in $\mathcal{R}^d$. Their PTAS also works for those infinite set of balls which can be represented implicitly in a specific way mentioned in their paper. Also it is worth it to mention that the bicriteria algorithm for $k$-means \cite{DMakarychevMSW15}, as mentioned before, is based on a local search method.
\subsection{Our Results and Techniques}
In this work we consider both SOS-FL and $k$-means. The main contribution of this work is that we are being able to handle the square of distances in an elegant way, which yields near optimal approximation bounds. This is in particular very interesting, as it gives a better understanding of the classical $k$-means problem, whose status in the plane has remained open for a long time. We design polynomial time approximation algorithms based on local search technique for both of these problems. Given an $\eps > 0$, the algorithm for SOS-FL yields a $(1+\eps)$-factor approximation. For $k$-means, the algorithm uses at most $(1+\eps)k$ centers and yields a solution whose cost is at most $(1+\eps)$ times the cost of an optimal $k$-means solution.

The algorithm and the analysis for both of the problems are similar. However, in case of $k$-means there are some more subtleties, which arise due to the limitation on the number of centers. In general, both of the algorithms are based on a local search method that allows swapping in and out of constant number of facilities or centers. Like the approaches in \cite{BhattiproluH14,ChanH12,Cohen-AddadM15,MustafaR09} we also use separators to prove the quality of the approximation. To be precise we use the separator from \cite{BhattiproluH14} which is most suitable for our purpose. We note that this separator itself gives a lot of ease in handling the square of distances. The separator is used repeatedly to partition the local and global optimal facilities simultaneously into constant size ``parts''. The rest of the analysis involves assignment of clients corresponding to each ``part'' to the global facilities corresponding to that ``part'' only or to some ``auxilliary'' points. Also one should be careful that a client should not be assigned to a point ``far'' away from it compare to its nearest local and global facility. The choice of the separator plays a crucial role to give a bound on this cost. From a very high level our approach is similar to the approach of Cohen-Addad and Mathieu \cite{Cohen-AddadM15} for clustering problems. But the details of the analysis are significantly different in places. For example, they use the dissection technique from \cite{KolliopoulosR07} as their separator and thus the assignment in their case is completely different and more complicated than ours. In this regard we would like to mention, that the dissection technique from \cite{KolliopoulosR07} or the quadtree based approach of Arora \cite{Arora98} are not flexible enough to handle the square of distances. The local search algorithm for SOS-FL and $k$-means are described in Section \ref{sec:FL} and Section \ref{sec:kmeans}, respectively.

\section{PTAS for Sum of Squares Facility Location}\label{sec:FL}
In this section we describe a simple local search algorithm for SOS-FL. We show that the solution returned by this algorithm is within $(1+O(\eps))$-factor of the optimal solution for any $\eps > 0$. Recall that in SOS-FL we are given a set $C$ of points in $\mathcal{R}^d$ and a real $f > 0$. Let $|C|=n$. For a point $p$ and a set $R$ of points, let $d(p,R) = \min_{q\in R} ||p-q||$.
\subsection{The Local Search Algorithm}

Fix an $\eps > 0$. The local search algorithm starts with a solution where one facility is placed at each client (see Algorithm \ref{alg:local}). Note that the cost of this solution is $nf$. Denote by $OPT$ the cost of any optimal solution. As $OPT \geq f$, the initial solution has cost at most $O(n\cdot OPT)$. In each iteration the algorithm looks for local improvement. The algorithm returns the current set of facilities if there is no such improvement. Notice that in line 2, we consider swaps with at most $\frac{c}{{\eps}^d}$ facilities, where $c$ is a constant. Fix $A=F\setminus F_1$, the facilities being swapped out. We are interested in finding a constant sized subset $B=F_1\setminus F$ such that $cost((F\setminus A)\cup B)$ is at most a certain quantity. Since this involves a search for a constant number of points in fixed dimension, it can be executed in polynomial time using standard techniques \cite{InabaKI94}.

\begin{algorithm}[]
     \caption{Local Search}
    \label{alg:local}
    \begin{algorithmic}[1]
        \REQUIRE A set of clients $C \subset \mathcal{R}^d$, a constant $\eps > 0$.
        \ENSURE A set of facilities $F$.
	\STATE $F \leftarrow$ the set of facilities with one facility at each client
 	\WHILE {$\exists$ a set $F_1$ s.t. $cost(F_1) < (1-\frac{1}{n})cost(F)$ and $|F_1\setminus F|+|F\setminus F_1|\leq \frac{c}{{\eps}^d}$}
	    \STATE $F \leftarrow F_1$
	\ENDWHILE
        \RETURN $F$
    \end{algorithmic}
\end{algorithm}

\subsection{Analysis of the Local Search Algorithm}
We analyze the local search algorithm using a partitioning scheme based on the separator theorem from \cite{BhattiproluH14}. To our surprise the analysis is simple. The idea is to partition the set of facilities in the local search solution and an optimal solution into parts of size $O(\frac{1}{{\eps}^d})$. Now for each such small part, we assign the clients corresponding to the local facilities of that part, either to the optimal facilities in that part or to the points belong to a special set. This yields a new solution, whose symmetric difference with the local search solution contains $O(\frac{1}{{\eps}^d})$ facilities. Thus using the local optimality criteria the cost of this new solution is not ``small'' compare to the local search solution. Then we combine the new solutions corresponding to the parts to give a bound on the cost of the local search solution. To start with we describe the separator theorem.

\subsubsection{Separator Theorem}
%
%
A ball $B$ is said to be stabbed by a point $p$ if $p \in B$. The following theorem is due to Bhattiprolu and Har-Peled \cite{BhattiproluH14} which shows the existence of a ``small'' point set (separator), that divides a given set of points into two in a ``balanced'' manner.

\begin{theorem}\label{th:bhatti}
 \cite{BhattiproluH14} (Separator Theorem) Let $X$ be a set of points in $\mathcal{R}^d$, and $\mu >0$ be an integer such that $|X|> \alpha\mu$, where $\alpha$ is a constant. There is an algorithm which can compute, in $O(|X|)$ expected time, a set $\Z$ of $O({\mu}^{1-\frac{1}{d}})$ points and a sphere $\mathcal{S}$ containing $\Theta(\mu)$ points of $X$ inside it, such that for any set $\mathcal{B}$ of balls stabbed by $X$, we have that every ball of $\mathcal{B}$ that intersects $\mathcal{S}$ is stabbed by a point of $\Z$.
 
\end{theorem}

The next corollary follows from Theorem \ref{th:bhatti} which will be very useful for the analysis of our local search algorithm. 

\begin{corollary}\label{cor:separator}
 Let $X$ be a set of points in $\mathcal{R}^d$,and $\mu >0$ be an integer such that $|X|> \alpha\mu$, where $\alpha$ is a constant. There is an algorithm which can compute, in $O(|X|)$ expected time, a set $\Z$ of $O({\mu}^{1-\frac{1}{d}})$ points and a ball $B$ containing $\Theta(\mu)$ points of $X$ in it, such that for any point $p \in \mathcal{R}^d$, $d(p,{\Z}) \leq \max\{d(p,X \setminus B), d(p,B \cap X)\}$.
\end{corollary}

\begin{proof}
We use the same Algorithm in Theorem \ref{th:bhatti} to compute the sphere $\mathcal{S}$ and the set $\Z$. Let $B$ be the ball that has $\mathcal{S}$ as its boundary. Now consider any point $p \in \mathcal{R}^d$. Let $p_1$ (resp. $p_2$) be a point in $B \cap X$ (resp. $X \setminus B$) nearest to $p$. If $p \in B$, the ball $B_1$ centered at $p$ and having radius $||p-p_2||$ must intersect the sphere $\mathcal{S}$, as $p_2 \notin B$. As $p_2$ stabs $B_1$, by Theorem \ref{th:bhatti} there is a point in $\Z$ that also stabs $B_1$. Hence $d(p,{\Z}) \leq ||p-p_2|| = d(p,X \setminus B)$. Similarly, if $p \notin B$, the ball $B_2$ centered at $p$ and having radius $||p-p_1||$ intersects the sphere $\mathcal{S}$, as $p_1 \in B$. As $B_2$ is stabbed by $p_1$, by Theorem \ref{th:bhatti} there is a point in $\Z$ that also stabs $B_2$. Hence $d(p,{\Z}) \leq ||p-p_1|| \leq d(p,B \cap X)$ and the corollary follows.
%
\end{proof}

The algorithm in Corollary \ref{cor:separator} will be referred to as the \textit{Separator algorithm}.
\subsubsection{The Partitioning Algorithm}
For the sake of analysis fix an optimal solution $\Op$. Let $\Lp$ be the solution computed by the local search algorithm. We design a procedure PARTITION$({\Lp},{\Op},\eps)$ which divide the set $\Lp \cup \Op$ into disjoint subsets of small size using the Separator algorithm (see Algorithm \ref{alg:part}). The procedure iteratively removes points from the set until the size of the set becomes less than or equal to $\alpha\mu$, where $\alpha$ is the constant in Corollary \ref{cor:separator}, and $\mu=\frac{\gamma}{{\eps}^d}$ for some constant $\gamma$. Next, we describe some important properties of this procedure which will be helpful to give an approximation bound on the cost of the local search solution. But before proceeding further we define some notation.

\begin{algorithm}[]
    \caption{PARTITION$({\Lp},{\Op},\eps)$}
    \label{alg:part}
    \begin{algorithmic}[1]
	\STATE $\mu=\frac{\gamma}{{\eps}^d}$, $i=1$, $L_1=\Lp$, $O_1=\Op$, ${\Z}_1=\phi$
	\WHILE {$|L_i\cup O_i \cup {\Z}_i| > \alpha\mu$}
	    \STATE Let $B_i,T_i$ be the ball and the point set computed by applying the Separator algorithm on the set $L_i\cup O_i\cup {\Z}_i$ with parameter $\mu$
	    \STATE ${\Lp}_i=L_i\cap B_i, {\Op}_i=O_i\cap B_i$
	    \STATE $L_{i+1}=L_i\setminus {\Lp}_i, O_{i+1} = O_i \setminus {\Op}_i$
	    \STATE ${\Z}_{i+1}=({\Z}_i\setminus B_i)\cup T_i$
	    \STATE $i=i+1$
	\ENDWHILE
	\STATE $I=i$
	\STATE Let $B_I$ be any ball that contains all the points in $L_I\cup O_I \cup {\Z}_I$
	\STATE $T_I = \phi$
    \end{algorithmic}
\end{algorithm}

Let $T=\cup_{i=1}^I T_i$ be the union of the point sets computed by the Separator algorithm in PARTITION$({\Lp},{\Op},\eps)$. 
Also let $C_l=\{p\in C| d(p,{\Lp}) \leq d(p,{\Op})\}$ and $C_o = C\setminus C_l$. Consider a point set $R \subset \mathcal{R}^d$. We denote the nearest neighbor voronoi diagram of $R$ by ${\V}_R$. For $p \in R$, let ${\V}_R(p)$ be the voronoi cell of $p$ in ${\V}_R$. Also let $C_R(p)={\V}_R(p)\cap C$, that is $C_R(p)$ is the set of clients that are contained in the voronoi cell of $p$ in ${\V}_R$. For $Q\subseteq R$, define $C_R(Q)={\cup}_{q\in Q} C_R(q)$.

Now consider a client $c$. Denote its nearest neighbor in $\Op$ (resp. $\Lp$) by $c(\Op)$ (resp. $c(\Lp)$). Also let $c_O=||c-c({\Op})||^2$ and $c_L=||c-c({\Lp})||^2$.

\begin{definition}
 An \textit{assignment} is a function that maps a set of clients to the set $\Lp\cup \Op\cup T$.
\end{definition}

Now with all these definitions we move on towards the analysis. We begin with the following observation. 

\begin{obs}\label{obs:cardTandcardproxy}
Consider the procedure PARTITION$({\Lp},{\Op},\eps)$. The following statements are true.
\begin{enumerate}
 \item $|{\Lp}_i\cup {\Op}_i\cup T_i \cup ({\Z}_i\cap B_i)|\leq\frac{\beta}{{\eps}^d}$ for a constant $\beta$ and $1\leq i\leq I$
 \item $I\leq {\eps}(|{\Lp}|+ |{\Op}|)/10$
 \item $|T|\leq {\eps}(|{\Lp}|+ |{\Op}|)/10$
 \item $\sum_{i=1}^I |T_i \cup ({\Z}_i\cap B_i)|\leq {\eps}(|{\Lp}|+ |{\Op}|)/5$
\end{enumerate}
\end{obs}

\begin{proof}
 1. Note that $|T_i| = O({\mu}^{1-\frac{1}{d}}) = O(\frac{1}{{\eps}^{d-1}})$ and $|{\Lp}_i\cup {\Op}_i\cup {\Z}_i| = \theta(\mu) = \theta(\frac{1}{{\eps}^d})$ for $1\leq i\leq I$. Thus there is a constant $\beta$ such that $|{\Lp}_i\cup {\Op}_i\cup T_i \cup ({\Z}_i\cap B_i)|\leq\frac{\beta}{{\eps}^d}$.
 \\2. As in each iteration we add $O({\mu}^{1-\frac{1}{d}})$ points and remove $\theta(\mu)$ points, the number of iterations $I = O(\frac{|{\Lp}|+ |{\Op}|}{\mu})=O({\eps}^d(\frac{|{\Lp}|+ |{\Op}|}{\gamma}))$. By choosing the constant $\gamma$ sufficiently large one can ensure that $I \leq {\eps}(|{\Lp}|+ |{\Op}|)/10$.
 \\3. $|T|=\sum_{i=1}^I |T_i| = O(\frac{|{\Lp}|+ |{\Op}|}{\mu}\cdot {\mu}^{1-\frac{1}{d}}) = O(\eps(|{\Lp}|+ |{\Op}|)/{\gamma}^{1-\frac{1}{d}}) \leq {\eps}(|{\Lp}|+ |{\Op}|)/10$, by choosing the value of $\gamma$ sufficiently large.
 \\4. Consider any point $p \in T_i$ for some $1\leq i\leq I$. If $p \in {\Z}_j\cap B_j$ for some $j > i$, then $p$ is removed in iteration $j$ and hence cannot appear in any other ${\Z}_t\cap B_t$ for $j+1 \leq t \leq I$. Thus $p$ can appear in at most two sets in the collection $\{T_1 \cup ({\Z}_1\cap B_1),\ldots,T_I \cup ({\Z}_I\cap B_I)\}$. It follows that $\sum_{i=1}^I |T_i \cup ({\Z}_i\cap B_i)|\leq 2\sum_{i=1}^I |T_i|=2|T|\leq {\eps}(|{\Lp}|+ |{\Op}|)/5$.
\end{proof}

The next lemma states the existence of a ``cheap'' assignment for any client $c$, such that its nearest neighbor $c(\Op)$ in $\Op$ is in ${\Op}_i$, its nearest neighbor $c(\Lp)$ in ${\Lp}$ is in ${\Lp}_j$ with $i< j$.

\begin{lemma}\label{lem:ginballj}
 Consider any client $c$, such that $c({\Op}) \in {\Op}_i$, $c({\Lp}) \in {\Lp}_j$ with $1 \leq i< j\leq I$. Also consider the sets ${T}_j$, ${\Z}_j$, and the ball $B_j$ computed by PARTITION$({\Lp},{\Op},\eps)$. There exists a point $p \in ({\Z}_j\cap B_j)\cup T_j$ such that $||c-p|| \leq \max\{||c-c({\Op})||,||c-c({\Lp})||\}$.
\end{lemma}

\begin{proof}
To prove this lemma, at first we prove the following claim.
 \begin{claim}\label{cl:qinz}
  For any $i+1 \leq t\leq j$, there exists a point $p \in {\Z}_t$ such that $||c-p|| \leq \max\{||c-c({\Op})||,||c-c({\Lp})||\}$.
 \end{claim}

 \begin{proof}
 We prove this claim using induction on the iteration number. In base case consider the iteration $i$. Let $X=L_i\cup O_i \cup {\Z}_i$. As $c({\Op}) \in B_i$ and $c({\Lp}) \in B_j$, $c({\Op}), c({\Lp}) \in X$. Now by Corollary \ref{cor:separator}, $d(c,{T_i}) \leq \max\{d(c,X \setminus B_i), d(c,B_i \cap X)\}$. As $c({\Op})$ is the nearest neighbor of $c$ in ${\Op}$ and $c({\Op}) \in B_i$, $d(c,B_i \cap X) \leq ||c-c({\Op})||$. Also $c({\Lp})$ is the nearest neighbor of $c$ in ${\Lp}$ and $c({\Lp}) \notin B_i$. Thus $d(c,X \setminus B_i) \leq ||c-c({\Lp})||$. Hence $d(c,{T_i}) \leq \max\{||c-c({\Op})||,||c-c({\Lp})||\}$. Let $p$ be the point in $T_i$ nearest to $c$. As $T_i \subseteq {\Z}_{i+1}$, $p \in {\Z}_{i+1}$ and the base case holds. 

 Now suppose the claim is true for any iteration $t < j-1 \leq I-1$. We show that the claim is also true for iteration $t+1$. By induction, there is a point $p\in {\Z}_{t+1}$ such that $||c-p|| \leq \max\{||c-c({\Op})||,||c-c({\Lp})||\}$. Now there can be two cases: (i) $p\notin B_{t+1}$, and (ii) $p\in B_{t+1}$. Consider the first case. In this case, by definition of ${\Z}_{t+2}$, $p\in {\Z}_{t+2}$ and the claim holds. Thus consider the second case. Let $X=L_{t+1}\cup O_{t+1} \cup {\Z}_{t+1}$. By Corollary \ref{cor:separator}, $d(c,{T_{t+1}}) \leq \max\{d(c,X \setminus B_{t+1}), d(c,B_{t+1} \cap X)\}$. As $p \in B_{t+1}\cap X$, $d(c,B_{t+1} \cap X) \leq ||c-p||\leq \max\{||c-c({\Op})||,||c-c({\Lp})||\}$. Now $c({\Lp}) \notin B_{t+1}$, as $t < j-1$. Also $c({\Lp})\in L_{t+1}\subseteq X$. Thus $d(c,X \setminus B_{t+1}) \leq ||c-c({\Lp})||$. It follows that $d(c,{T_{t+1}}) \leq \max\{||c-c({\Op})||,||c-c({\Lp})||\}$. Let $q$ be the point in $T_{t+1}$ nearest to $c$. As $T_{t+1} \subseteq {\Z}_{t+2}$, $q \in {\Z}_{t+2}$ and the claim holds also for this case.
 \end{proof}

 Consider the iteration $j$. From Claim \ref{cl:qinz} it follows that there exists a point $p \in {\Z}_j$ such that $||c-p|| \leq \max\{||c-c({\Op})||,||c-c({\Lp})||\}$. Thus if $p \in B_{j}$, then $p \in {\Z}_{j}\cap B_{j}$, and we are done. Note that the way $B_I$ is chosen, ${\Z}_{I} \subseteq B_I$. Thus in case $j =I$, $p \in {\Z}_{j}\cap B_{j}$. Hence consider the case when $j \neq I$ and $p \notin B_{j}$. Let $X=L_{j}\cup O_{j} \cup {\Z}_{j}$. As $p \in {\Z}_{j}$, $p \in X$. Also by Corollary \ref{cor:separator}, $d(c,{T_{j}}) \leq \max\{d(c,X \setminus B_{j}), d(c,B_{j} \cap X)\}$. As $c({\Lp}) \in B_{j}\cap X$, $d(c,B_{j} \cap X) \leq ||c-c({\Lp})||$. Now $p \in X \setminus B_{j}$. Thus $d(c,X \setminus B_{j}) \leq ||c-p|| \leq \max\{||c-c({\Op})||,||c-c({\Lp})||\}$. Hence $d(c,{T_{j}}) \leq \max\{||c-c({\Op})||,||c-c({\Lp})||\}$ and the lemma follows. 
\end{proof}

Now we extend Lemma \ref{lem:ginballj} for any client whose nearest neighbor in ${\Lp}$ is in ${\Lp}_j$, but the nearest neighbor in ${\Op}$ is not in ${\Op}_j$, where $1\leq j\leq I$.

\begin{lemma}\label{lem:gnotinballi}
 Consider the sets ${T}_j$, ${\Z}_j$ and the ball $B_j$ computed by PARTITION$({\Lp},{\Op},\eps)$, where $1\leq j\leq I$. There is an assignment $g$ of the clients in $C_{\Lp}({\Lp}_j)\setminus C_{\Op}({\Op}_j)$ to $T_j \cup ({\Z}_j\cap B_j)$ with the following properties: 
 \begin{enumerate}
  \item for $c\in (C_{\Lp}({\Lp}_j)\setminus C_{\Op}({\Op}_j))\cap C_l$, ${||c-g(c)||}^2 \leq c_O$;
  \item for $c\in (C_{\Lp}({\Lp}_j)\setminus C_{\Op}({\Op}_j))\cap C_o$, ${||c-g(c)||}^2 \leq c_L$.
 \end{enumerate}
\end{lemma}

\begin{proof}
 We show how to construct the assignment $g$ for each client in $C_{\Lp}({\Lp}_j)\setminus C_{\Op}({\Op}_j)$. Consider any client $c \in C_{\Lp}({\Lp}_j)\setminus C_{\Op}({\Op}_j)$. Let ${\Op}_i$ be the subset of ${\Op}$ that contains $c({\Op})$. If $j$ is equal to $I$, then $i < j$. Otherwise, there could be two cases: (i) $i < j$, and (ii) $i > j$. Consider the case when $i < j$ for $1\leq j\leq I$. By Lemma \ref{lem:ginballj}, there is a point $p \in ({\Z}_j\cap B_j)\cup T_j$ such that $||c-p|| \leq \max\{||c-c({\Op})||,||c-c({\Lp})||\}$. Let $g(c)$ be $p$ in this case. Now consider the case when $i > j$ such that $j < I$. Let $X=L_{j}\cup O_{j} \cup {\Z}_{j}$. By Corollary \ref{cor:separator}, $d(c,{T_{j}}) \leq \max\{d(c,X \setminus B_{j}), d(c,B_{j} \cap X)\}$. As $c({\Lp}) \in B_{j}\cap X$, $d(c,B_{j} \cap X) \leq ||c-c({\Lp})||$. Now note that $c({\Op}) \in X$. Also $c({\Op}) \notin B_j$, as $c({\Op})\in {\Op}_i$ and $i > j$. Thus $c({\Op}) \in X \setminus B_{j}$ and $d(c,X \setminus B_{j}) \leq ||c-c({\Op})||$. Hence $d(c,{T_{j}}) \leq \max\{||c-c({\Op})||,||c-c({\Lp})||\}$. Let $g(c)$ be the point in $T_j$ nearest to $c$ in this case. 
 
 In both cases $||c-g(c)||\leq \max\{||c-c({\Lp})||,||c-c({\Op})||\}$. If $c \in C_l$, ${||c-g(c)||}^2 \leq {||c-c({\Op})||}^2=c_O$. Otherwise, $c \in C_o$, and thus ${||c-g(c)||}^2 \leq {||c-c({\Lp})||}^2= c_L$. Hence the lemma holds.
\end{proof}

\subsubsection{Approximation Bound}
%

The next lemma gives an upper bound on the quality of the local search solution.

\begin{lemma}
 $cost({\Lp}) \leq (1+O(\eps)) cost({\Op})$.
\end{lemma}

\begin{proof}
Fix an iteration $i$, where $1\leq i\leq I$. Consider the solution $S_i=({\Lp}\setminus{\Lp}_i)\cup {\Op}_i\cup T_i \cup ({\Z}_i\cap B_i)$. By Observation \ref{obs:cardTandcardproxy}, $|{\Lp}_i\cup {\Op}_i\cup T_i \cup ({\Z}_i\cap B_i)|\leq\frac{\beta}{{\eps}^d}$. Thus $|{\Lp}\setminus S_i|+|S_i\setminus {\Lp}|\leq\frac{\beta}{{\eps}^d}$. By choosing the constant $c$ in Algorithm \ref{alg:local} sufficiently large, one can ensure that $\beta \leq c$. Hence due to the local optimality condition in Algorithm \ref{alg:local} it follows that, 

\begin{dmath}\label{in:costsi} 
cost(S_i)\geq (1-\frac{1}{n})cost(\Lp)
\end{dmath}
To argue about the cost of $S_i$ we use an assignment of the clients to the facilities in $S_i$. Consider a client $c$. There can be three cases: (i) $c$ is nearer to a facility of ${\Op}_i$ than the facilities in $({{\Op}\setminus {\Op}_i})\cup ({\Lp}\setminus {\Lp}_i)$, that is, $c \in C_{\Op}({\Op}_i)\cap (C_o\cup C_{\Lp}({\Lp}_i))$, (ii) $c$ is not in $C_{\Op}({\Op}_i)$ and $c$ is nearer to a facility of ${\Lp}_i$ than the facilities in ${\Lp}\setminus {\Lp}_i$, that is, $c \in C_{\Lp}({\Lp}_i) \setminus C_{\Op}({\Op}_i)$, (iii) $c$ does not appear in case (i) and (ii), that is, $c$ is not in the union of $C_{\Op}({\Op}_i)\cap (C_o\cup C_{\Lp}({\Lp}_i))$ and $C_{\Lp}({\Lp}_i) \setminus C_{\Op}({\Op}_i)$. Let $R_i$ be the set of the clients that appear in case (iii). Note that a client cannot appear in both cases (i) and (ii), as the sets $C_{\Op}({\Op}_i)\cap (C_o\cup C_{\Lp}({\Lp}_i))$ and $C_{\Lp}({\Lp}_i) \setminus C_{\Op}({\Op}_i)$ are disjoint. Also note, that if $c\in C_{\Lp}({\Lp}_i)$, $c$ must appear in case (i) or (ii). Thus if $c$ is corresponding to case (iii), its nearest neighbor $c({\Lp})$ in ${\Lp}$ must be in ${\Lp}\setminus {\Lp}_i$.

Now we describe the assignment. Note that we can assign the clients only to the facilities in $S_i=({\Lp}\setminus{\Lp}_i)\cup {\Op}_i\cup T_i \cup ({\Z}_i\cap B_i)$. For a client of type (i), assign it to a facility in ${\Op}_i$ nearest to it. For a client of type (ii), use the assignment $g$ in Lemma \ref{lem:gnotinballi} to assign it to a point in $T_i \cup ({\Z}_i\cap B_i)$. For a client of type (iii), assign it to a facility in ${\Lp}\setminus {\Lp}_i$ nearest to it. Thus by Inequality \ref{in:costsi},
\begin{multline}
 |({\Lp}\setminus{\Lp}_i)\cup {\Op}_i\cup T_i \cup ({\Z}_i\cap B_i)|f+\sum_{c\in C_{\Op}({\Op}_i)\cap (C_o\cup C_{\Lp}({\Lp}_i))} c_O+\sum_{c\in C_{\Lp}({\Lp}_i) \setminus C_{\Op}({\Op}_i)} {||c-g(c)||}^2+\\ \sum_{c\in R_i} c_L \geq (1-\frac{1}{n})(|{\Lp}|f+\sum_{c\in C} c_L)
\end{multline}
By Lemma \ref{lem:gnotinballi}, for a client $c$ in $(C_{\Lp}({\Lp}_i) \setminus C_{\Op}({\Op}_i))\cap C_o$, ${||c-g(c)||}^2$ is at most $c_L$; for a client $c$ in $(C_{\Lp}({\Lp}_i) \setminus C_{\Op}({\Op}_i))\cap C_l$, ${||c-g(c)||}^2$ is at most $c_O$. It follows that,
\begin{multline}
 |{\Op}_i\cup T_i \cup ({\Z}_i\cap B_i)|f+\sum_{c\in C_{\Op}({\Op}_i)\cap (C_o\cup C_{\Lp}({\Lp}_i))} (c_O-c_L)+\sum_{c\in (C_{\Lp}({\Lp}_i) \setminus C_{\Op}({\Op}_i))\cap C_o} (c_L-c_L)+\\ \sum_{c\in (C_{\Lp}({\Lp}_i) \setminus C_{\Op}({\Op}_i))\cap C_l} (c_O-c_L) \geq -\frac{1}{n}cost({\Lp})+|{\Lp}_i|f
\end{multline}
\begin{multline}
 \Rightarrow |{\Op}_i\cup T_i \cup ({\Z}_i\cap B_i)|f+\sum_{c\in C_{\Op}({\Op}_i)\cap (C_o\cup C_{\Lp}({\Lp}_i))} (c_O-c_L)+ \sum_{c\in (C_{\Lp}({\Lp}_i) \setminus C_{\Op}({\Op}_i))\cap C_l} (c_O-c_L)\\ \geq -\frac{1}{n}cost({\Lp})+|{\Lp}_i|f
\end{multline}
\begin{align*}
 \Rightarrow |{\Op}_i\cup T_i \cup ({\Z}_i\cap B_i)|f+\sum_{c\in C_{{\Op}\cup {\Lp}} ({\Op}_i\cup {\Lp}_i)} (c_O-c_L) \geq -\frac{1}{n}cost({\Lp})+|{\Lp}_i|f
\end{align*}

The last inequality follows by noting that the union of $C_{\Op}({\Op}_i)\cap (C_o\cup C_{\Lp}({\Lp}_i))$ and $(C_{\Lp}({\Lp}_i) \setminus C_{\Op}({\Op}_i))\cap C_l$ is equal to the set $C_{{\Op}\cup {\Lp}} ({\Op}_i\cup {\Lp}_i)$. Summing over all $i$ we get,

\begin{multline}
\sum_{i=1}^I |{\Op}_i|f+\sum_{i=1}^I |T_i \cup ({\Z}_i\cap B_i)|f+\sum_{i=1}^I \sum_{c\in C_{\Op\cup \Lp}({\Op}_i\cup {\Lp}_i)} (c_O-c_L) \geq -O(\eps) cost({\Lp})+\sum_{i=1}^I |{\Lp}_i|f
\end{multline}

Note that $I=O({\eps}(|{\Lp}|+ |{\Op}|))=O({\eps}n)$ by Observation \ref{obs:cardTandcardproxy}. Now again by Observation \ref{obs:cardTandcardproxy}, $\sum_{i=1}^I |T_i \cup ({\Z}_i\cap B_i)|=O({\eps}(|{\Lp}|+ |{\Op}|))$. Hence we get,

\begin{align*}
|{\Op}|f+O(\eps(|{\Lp}|+ |{\Op}|))f+\sum_{c\in C} (c_O-c_L) \geq -O(\eps) cost({\Lp})+|{\Lp}|f
\end{align*}
\begin{align*}
 \Rightarrow cost({\Lp}) \leq (1+O(\eps)) cost({\Op})
\end{align*}
This completes the proof of the lemma.
\end{proof}

As mentioned before for fixed $d$, the running time of Algorithm \ref{alg:local} is polynomial and hence we have established the following theorem.
\begin{theorem}\label{th:flptas}
 There is a local search algorithm for SOS-FL which yields a PTAS.
\end{theorem}

\section{Bi-criteria Approximation scheme for $k$-means}\label{sec:kmeans}
In this section we describe a local search algorithm for $k$-means which uses $(1+O(\eps))k$ centers and yields a solution whose cost is at most $(1+O(\eps))$ times the cost of an optimal $k$-means solution. The local search algorithm and its analysis are very similar to the ones for SOS-FL. Recall that in $k$-means we are given a set $P$ of $n$ points in $\mathcal{R}^d$ and an integer $k > 0$.
\subsection{The Local Search Algorithm}
Fix an $\eps > 0$. 
The local search algorithm starts with the solution computed by the $9+\eps$ factor approximation algorithm in \cite{KanungoMNPSW04} (see Algorithm \ref{alg:kmeans}). 
Upon termination, the locally optimal solution $K$ has exactly $(1+5\eps)k$ centers. Using a standard argument like in \cite{InabaKI94}, one can show that this algorithm runs in polynomial time.

\begin{algorithm}[]
     \caption{Local Search}
    \label{alg:kmeans}
    \begin{algorithmic}[1]
        \REQUIRE A set of points $P \subset \mathcal{R}^d$, an integer $k$, a constant $\eps > 0$.
        \ENSURE A set of centers.
	\STATE $K \leftarrow$ the solution returned by the algorithm in \cite{KanungoMNPSW04}
	\STATE Add arbitrary centers to $K$ to ensure $|K|=(1+5\eps)k$	    
 	\WHILE {$\exists$ a set $K_1$ s.t. $|K_1| \leq (1+5\eps)k$, $cost(K_1) < (1-\frac{1}{n})cost(K)$ and $|K_1\setminus K|+|K\setminus K_1|\leq \frac{c}{{\eps}^{2d}}$}
	  \STATE $K\leftarrow K_1$
	  \STATE If needed, add arbitrary centers to $K$ to ensure $|K|=(1+5\eps)k$	    
	\ENDWHILE
        \RETURN $K$
    \end{algorithmic}
\end{algorithm}

\subsection{Analysis of the Local Search Algorithm}
Let $\Lp$ be the solution computed by Algorithm \ref{alg:kmeans}. For the sake of analysis fix an optimal solution $\Op$. We use the procedure PARTITION$({\Lp},{\Op},\eps)$ to 
compute the sets ${\Lp}_i$, ${\Op}_i$, ${\Z}_i\cap B_i$ and $T_i$ for $1\leq i\leq I$. We use the same $\mu=\frac{\gamma}{{\eps}^d}$ in this procedure. 
We note, that Observation \ref{obs:cardTandcardproxy}, Lemma \ref{lem:gnotinballi}, and Lemma \ref{lem:ginballj} hold in this case also, as they directly follow from the PARTITION procedure, which works on any two input sets of points designated by $\Lp$ and $\Op$. Let $R_i={\Lp}_i\cup {\Op}_i\cup T_i \cup ({\Z}_i\cap B_i)$ for $1\leq i\leq I$. Note, by Observation \ref{obs:cardTandcardproxy}, that $|R_i| \leq \frac{\beta}{{\eps}^d}$. Also note that ${\Z}_i\cap B_i \subseteq T$ for each $1\leq i\leq I$, where $T=\cup_{i=1}^I T_i$. Now we use the following lemma to group the balls returned by PARTITION into groups of ``small'' size. This lemma is similar to the Balanced Clustering Lemma in \cite{Cohen-AddadM15}.

\begin{lemma}\label{lem:cluster}
 Consider the collection $R = \{R_1,\ldots,R_I\}$ of sets with $R_j = {\Lp}_j\cup {\Op}_j\cup T_j \cup ({\Z}_j\cap B_j)$ and $|R_j| \leq \frac{\beta}{{\eps}^d}$ for $1\leq j\leq I$, where $\beta$ is the constant in Observation \ref{obs:cardTandcardproxy}. There exists a collection $\mathcal{P}=\{{P}_1,\ldots, {P}_p\}$, with $P_i \subseteq R$ for $1\leq i\leq p$, $P_i\cap P_j=\phi$ for any $1\leq i < j\leq p$, and $\cup_{i=1}^p P_i =R$, which satisfies the following properties:
 \begin{enumerate}
  \item  $|P_i| \leq \frac{2\beta}{{\eps}^d}$ for $1\leq i\leq p$; 
  \item $\sum_{R_j \in P_i} |{\Lp}\cap R_j|\geq \sum_{R_j \in P_i} |({\Op}\cup T)\cap R_j|$ for $1\leq i\leq p$.
 \end{enumerate}
\end{lemma}

We note that ${\Lp}\cap R_j={\Lp}_j$ and $({\Op}\cup T)\cap R_j={\Op}_j\cup T_j \cup ({\Z}_j\cap B_j)$ for $1\leq j\leq I$. Before proving the lemma we use it to get an approximation bound on the quality of the local search solution. 

\begin{lemma}
 $cost({\Lp}) \leq (1+O(\eps)) cost({\Op})$.
\end{lemma}

\begin{proof}
Consider the collection $\mathcal{P}$ as mentioned in Lemma \ref{lem:cluster}. Also consider any element $J$ of $\mathcal{P}$. Let ${\Lp}_J=\cup_{R_i \in J} ({\Lp}\cap R_i)$, ${\Op}_J=\cup_{R_i \in J} ({\Op}\cap R_i)$, and $T_J = \cup_{R_i \in J} ({T}\cap R_i)$. Now consider the solution $S_J = ({\Lp}\setminus {\Lp}_J)\cup {\Op}_J\cup T_J$. By Lemma \ref{lem:cluster}, and using the fact that ${\Lp}\cap R_i$ and ${\Lp}\cap R_j$ are disjoint for $i\neq j$, $|{\Lp}_J|=\sum_{R_j \in J} |{\Lp}\cap R_j|\geq \sum_{R_j \in J} |({\Op}\cup T)\cap R_j| \geq |{\Op}_J\cup T_J|$. Thus by definition of $S_J$, $|S_J| \leq |{\Lp}|= (1+5\eps)k$. Now $J$ contains at most $\frac{2\beta}{{\eps}^d}$ sets and thus $|{\Lp}\setminus S_J|+|S_J\setminus {\Lp}|\leq |{\Lp}_J\cup {\Op}_J\cup T_J|\leq\frac{2\beta}{{\eps}^{d}}\frac{\beta}{{\eps}^{d}}\leq \frac{2{\beta}^2}{{\eps}^{2d}}$. By choosing the constant $c$ in Algorithm \ref{alg:kmeans} sufficiently large, one can ensure that $2{\beta}^2 \leq c$. Hence due to the local optimality condition in Algorithm \ref{alg:kmeans} it follows that, 
\begin{dmath}\label{in:costSJ}
cost(S_J)\geq (1-\frac{1}{n})cost({\Lp})
\end{dmath}
To argue about the cost of $S_J$ we use an assignment of the clients to the facilities in $S_J$. Consider a client $c$. There can be three cases: (i) $c$ is nearer to a facility of ${\Op}_J$ than the facilities in $({\Op}\setminus {\Op}_J)\cup ({\Lp}\setminus {\Lp}_J)$, that is, $c \in C_1 = C_{\Op}({\Op}_J)\cap (C_o\cup C_{\Lp}({\Lp}_J))$, (ii) $c$ is not in $C_{\Op}({\Op}_J)$ and $c$ is nearer to a facility of ${\Lp}_J$ than the facilities in ${\Lp}\setminus {\Lp}_J$, that is, $c \in C_2 = C_{\Lp}({\Lp}_J) \setminus C_{\Op}({\Op}_J)$, (iii) $c$ does not appear in case (i) and (ii), that is, $c \in C_3 = C \setminus (C_1\cup C_2)$. Note that if $c\in C_{\Lp}({\Lp}_J)$, $c$ must appear in case (i) or (ii). Thus if $c$ is corresponding to case (iii), its nearest neighbor in ${\Lp}$ must be in ${\Lp}\setminus {\Lp}_J$. Also note that for a client $c \in C_2$, it should be the case that $c({\Lp}) \in {\Lp}_i\subseteq {\Lp}_J$ and $c({\Op}) \in {\Op}_j \subseteq ({\Op}\setminus {\Op}_J)$ for some $1\leq i\neq j\leq I$. Thus Lemma \ref{lem:gnotinballi} is applicable for $c$, and $g(c) \in T_i \cup ({\Z}_i\cap B_i) \subseteq T_J$.
 
Now we describe the assignment. For a client in $C_1$, assign it to a facility in ${\Op}_J$ nearest to it. For a client $c \in C_2$, use the assignment $g$ in Lemma \ref{lem:gnotinballi}. For a client in $C_3$, assign it to a facility in ${\Lp}\setminus {\Lp}_J$ nearest to it. Thus by Inequality \ref{in:costSJ},

\begin{align*}
 \sum_{c\in C_1} c_O+\sum_{c\in C_2} {||c-g(c)||}^2+ \sum_{c\in C_3} c_L \geq (1-\frac{1}{n})\sum_{c\in C} c_L
\end{align*}
\begin{align*}
 \Rightarrow \sum_{c\in C_1} (c_O-c_L)+\sum_{c\in (C_2\cap C_o)} (c_L-c_L)+ \sum_{c\in (C_2\cap C_l)} (c_O-c_L) \geq -\frac{1}{n}cost({\Lp})
\end{align*}
\begin{align*}
 \Rightarrow \sum_{c\in C_1} (c_O-c_L)+ \sum_{c\in (C_2\cap C_l)} (c_O-c_L) \geq -\frac{1}{n}cost({\Lp})
\end{align*}
\begin{align*}
 \Rightarrow \sum_{c\in C_{{\Op}\cup {\Lp}} ({\Op}_J\cup {\Lp}_J)} (c_O-c_L) \geq -\frac{1}{n}cost({\Lp})
\end{align*}

The last inequality follows by noting that $C_1 \cup (C_2\cap C_l) = C_{{\Op}\cup {\Lp}} ({\Op}_J\cup {\Lp}_J)$. Also by Observation \ref{obs:cardTandcardproxy}, it follows that $|\mathcal{P}|\leq I=O({\eps}(|{\Lp}|+ |{\Op}|))=O({\eps}k)$. Thus by summing over all $J \in \mathcal{P}$ we get,

\begin{align*}
\sum_{J\in \mathcal{P}} \sum_{c\in C_{\Op\cup \Lp}({\Op}_J\cup {\Lp}_J)} (c_O-c_L) \geq -O(\eps) cost({\Lp})
\end{align*}
\begin{align*}
 \Rightarrow \sum_{c\in C} (c_O-c_L) \geq -O(\eps) cost({\Lp})
\end{align*}
\begin{align*}
 \Rightarrow cost({\Lp}) \leq (1+O(\eps)) cost({\Op})
\end{align*}
\end{proof}

As mentioned before for fixed $d$, the running time of Algorithm \ref{alg:local} is polynomial and hence we have established the following theorem.
\begin{theorem}\label{th:kmeans}
 There is a polynomial time local search algorithm for $k$-means that uses $(1+O(\eps))k$ facilities and returns a solution with cost at most $(1+O(\eps))$ times the cost of the optimal $k$-means solution.
\end{theorem}

Next we complete the proof of Lemma \ref{lem:cluster}.
\subsubsection{Proof of Lemma \ref{lem:cluster}}
For simpicity of exposition we further define some notations. For a set $r \subseteq {\Lp} \cup {\Op}\cup T$, let $u(r) = |{\Lp}\cap r|-|({\Op} \cup T)\cap r|$. For a collection $\Psi$ of sets, let $u(\Psi) = \sum_{r \in \Psi} |{\Lp}\cap r|-|({\Op} \cup T)\cap r|$. Using these notations we rewrite the statement of Lemma \ref{lem:cluster} as following. 
\begin{lemma}
 Consider the collection $R = \{R_1,\ldots,R_I\}$ of sets with $R_j = {\Lp}_j\cup {\Op}_j\cup T_j \cup ({\Z}_j\cap B_j)$ and $|R_j| \leq \frac{\beta}{{\eps}^d}$ for $1\leq j\leq I$, where $\beta$ is the constant in Observation \ref{obs:cardTandcardproxy}. There exists a collection $\mathcal{P}=\{{P}_1,\ldots, {P}_p\}$, with $P_i \subseteq R$ for $1\leq i\leq p$, $P_i\cap P_j=\phi$ for any $1\leq i < j\leq p$, and $\cup_{i=1}^p P_i =R$, which satisfies the following properties:
 \begin{enumerate}
  \item  $|P_i| \leq \frac{2\beta}{{\eps}^d}$ for $1\leq i\leq p$, where $\beta$ is the constant in Observation \ref{obs:cardTandcardproxy}; 
  \item $u(P_i)\geq 0$ for $1\leq i\leq p$.
 \end{enumerate} 
\end{lemma}

\begin{proof}
 Note that for each $j$, $u(R_j) \in [-\frac{\beta}{{\eps}^d},\frac{\beta}{{\eps}^d}]$, as $|R_j|\leq\frac{\beta}{{\eps}^d}$. Now by Observation \ref{obs:cardTandcardproxy}, $\sum_{j=1}^I |T_j \cup ({\Z}_j\cap B_j)|\leq \eps(|{\Lp}|+|{\Op}|)/5\leq \eps((1+5\eps)k+k)/5\leq 2\eps k$. Thus,
 
 \begin{align*}
 u(R) = |{\Lp}|-|{\Op}|-\sum_{j=1}^I |T_j \cup ({\Z}_j\cap B_j)| \geq (1+5\eps)k-k-2\eps k \geq 3\eps k.
 \end{align*}
 Now we show the construction of the collection $\mathcal{P}$. For any $j$, if $u(R_j)$ equals $0$, we add $\{R_j\}$ to $\mathcal{P}$ as an element. Note that such an element satisfies the desired properties. Now consider all the sets $R_j \in R$ such that $|u(R_j)| \geq 1$. Denote by $R'$ the collection of such sets. Note that $u(R')=u(R)$. We process $R'$ using the following construction.
 
 The construction is shown as Algorithm \ref{alg:cluster}. In each iteration of the outer while loop in line 2, we remove at most $l = \frac{2\beta}{{\eps}^d}$ sets from $R'$ and add the collection $\Psi$ of these sets to $\mathcal{P}$ as an element. These $l$ sets are chosen carefully so that $u(\Psi)$ is non-negative. To ensure this, at first $l/2$ sets are chosen to get the collection $\Psi'$ such that $-l/2 \leq u(\Psi') \leq l/2$. Then we add at most $l/2$ more sets $\{r\}$ with $u(r) > 0$ (while loop in lines $17$-$19$) to obtain a collection $\Psi$ with $u(\Psi) \geq 0$. Assuming the loop invariant $u(R') \geq 0$ at the beginning of each iteration of the outer while loop, such $r$ must exist. Later we will argue that this loop invariant holds. This ensures that the algorithm exits the while loop in lines $17$-$19$ with a $\Psi$ such that $u(\Psi) \geq 0$. Also we stop the addition of sets as soon as $u(\Psi)$ becomes non-negative. This ensures that $u(\Psi) \leq \frac{\beta}{{\eps}^d}$. Note that $|\Psi| \leq l = \frac{2\beta}{{\eps}^d}$. Thus $\Psi$ satisfies all the desired properties. 
 \begin{algorithm}[]
  \caption{}
  \label{alg:cluster}
  \begin{algorithmic}[1]
    \STATE $l = \frac{2\beta}{{\eps}^d}$
    \WHILE {$|R'| > l$}
      \STATE $\Psi'\leftarrow \{r\}$, where $r$ is any element in $R'$ with $u(r)>0$; $R'\leftarrow R'\setminus \{r\}$
      \FOR {$i= 1$ to $l/2$ }
	\IF {$u(\Psi') \geq 0$}
	  \IF {$u(r) > 0$ for each $r \in R'$}
	    \STATE Add $\Psi'$ to $\mathcal{P}$
	    \FOR {\textbf{each} $r \in R'$}
	      \STATE Add $\{r\}$ to $\mathcal{P}$; $R'\leftarrow R'\setminus \{r\}$
	    \ENDFOR
	    \RETURN
	  \ENDIF
	  \STATE $r\leftarrow$ any element in $R'$ with $u(r)<0$
	  \STATE $\Psi' \leftarrow \Psi'\cup \{r\}$; $R'\leftarrow R'\setminus \{r\}$
	\ELSIF {$u(\Psi') < 0$} 
	  \STATE $r\leftarrow$ any element in $R'$ with $u(r)>0$
	  \STATE $\Psi' \leftarrow \Psi'\cup \{r\}$; $R'\leftarrow R'\setminus \{r\}$
	\ENDIF	
      \ENDFOR
      \STATE $\Psi \leftarrow \Psi'$	
      \WHILE {$u(\Psi) < 0$}
	\STATE $r\leftarrow$ any element in $R'$ with $u(r)>0$
	\STATE $\Psi \leftarrow \Psi\cup \{r\}$; $R'\leftarrow R'\setminus \{r\}$
      \ENDWHILE
      \STATE Add $\Psi$ to $\mathcal{P}$
    \ENDWHILE
    \STATE Add $R'$ to $\mathcal{P}$
  \end{algorithmic}
\end{algorithm}

 Now consider the selection of the $l/2$ sets of $\Psi'$. We select these sets sequentially, one in each iteration of the for loop in lines 4-15. Consider a particular iteration of this for loop. There can be two cases: (i) $u(\Psi') \geq 0$, and (ii) $u(\Psi') < 0$. In case (i) if there is a set $r$ with $u(r) < 0$, we choose it. If there is no such set $r$, we add $\Psi'$ to $\mathcal{P}$ and for any set $r \in R'$, $\{r\}$ is added to $\mathcal{P}$ as an element. The algorithm terminates. In case (ii) we choose a set $r$ with $u(r) > 0$. Assuming the loop invariant $u(R') \geq 0$ at the beginning of each iteration of the outer while loop, such an $r$ must exist. Hence in both cases we can ensure that at the end of each iteration of the for loop in lines $4$-$15$ $u(\Psi')\in [-\frac{\beta}{{\eps}^d},\frac{\beta}{{\eps}^d}]$. 
 
 Let $M$ denote the number of iterations of the outer while loop. Then in each step except the last, we remove at least $l/2$ sets from $R'$. Since $R'$ has at most $I$ sets initially, $(M-1)\frac{l}{2} \leq I \Rightarrow M \leq \frac{2I}{l}+1$.
 
 Now we argue that after iteration $0\leq j\leq M$, 
 
\begin{dmath}\label{in:cluster} 
u(R') \geq (\frac{2I}{l}+1)\frac{l}{2}-j\frac{l}{2}
\end{dmath}

Since $(\frac{2I}{l}+1)\frac{l}{2}-j\frac{l}{2}\geq (\frac{2I}{l}+1)\frac{l}{2}-M\frac{l}{2} \geq 0$, this would imply $u(R')\geq 0$ after iteration $j\leq M$. This establishes the loop invariant and also shows that the set $R'$ added to $\mathcal{P}$ in line 21 has $u(R')\geq 0$, completing the proof of the lemma. 

We now show ($\ref{in:cluster}$) by induction. The inequality is true for $j=0$, since before iteration 1, 
\begin{align*}
u(R')=u(R)\geq 3\eps k\geq I+\frac{\beta}{{\eps}^d} = (\frac{2I}{l}+1)\frac{l}{2}
\end{align*}

Consider a $j$ such that $1\leq j \leq M$ and suppose ($\ref{in:cluster}$) is true after iteration $j-1$. Then at the beginning of iteration $j$, we have $u(R') \geq (\frac{2I}{l}+1)\frac{l}{2}-(j-1)\frac{l}{2}$. If the condition in line 6 is true in iteration $j$, then 
the algorithm terminates.
Since $R'$ becomes empty after this iteration, ($\ref{in:cluster}$) trivially holds. If in iteration $j$ we add $\Psi$ to $\mathcal{P}$ in line 20, then $u(\Psi) \leq \frac{\beta}{{\eps}^d} = \frac{l}{2}$. Thus, after iteration $j$,

\begin{align*}
u(R')\geq (\frac{2I}{l}+1)\frac{l}{2}-(j-1)\frac{l}{2}-\frac{l}{2} = (\frac{2I}{l}+1)\frac{l}{2}-j\frac{l}{2}. 
\end{align*}
\end{proof}
\bibliographystyle{plain}
\bibliography{facility}

\appendix

\end{document}